\documentclass[11pt]{article}
\usepackage{latexsym}
\usepackage[dvips]{graphics,color}
\usepackage{graphpap}
\usepackage{alltt}
\usepackage{graphicx}
\usepackage{amsfonts}
\usepackage{amsmath}
\usepackage{psfrag}
\newtheorem{theorem}{Theorem}[section]

\newtheorem{proposition}[theorem]{Proposition}
\newtheorem{proof}[theorem]{Proof}
\newtheorem{definition}[theorem]{Definition}
\newtheorem{corollary}[theorem]{Corollary}
\newcommand{\beq}{\begin{equation}}
\newcommand{\feq}[1]{\label{#1} \end{equation}}
\newcommand{\beqr}{\begin{eqnarray}}
\newcommand{\feqr}{\end{eqnarray}}
\def\non{\nonumber}

\newcommand{\rf}[1]{(\ref{#1})}

\definecolor{red}{rgb}{1,0,0}
\DeclareFontFamily{U}{eufm}{}
\DeclareFontShape{U}{eufm}{m}{n}{<->eufm10}{}
\DeclareSymbolFont{mcy}{U}{eufm}{m}{n}
\DeclareMathSymbol{\Hr}{\mathord}{mcy}{"58}

\def\np#1#2#3{Nucl. Phys. #1, {B#2,} #3}

\def\jhp#1#2#3{JHEP #1, #2, #3}

\def\tams#1#2#3{Trans. Amer. Math. Soc. #1, #2, #3}

\def\cm#1#2#3{Coll. Math. #1, {LXVII#2,} #3}

\usepackage[active]{srcltx}
\begin{document}

\begin{center}


{\Large \bf A note on doubly warped product spaces}\\
[4mm]

\large{Agapitos N. Hatzinikitas} \\ [5mm]

{\small Department of Mathematics, \\ 
University of Aegean, \\
School of Sciences, \\
Karlovasi, 83200\\
Samos Greece \\
E-mail: ahatz@aegean.gr}\\ [5mm]

\end{center}

\begin{abstract}
We present a coordinate free approach to derive curvature formulas for pseudo-Riemannian doubly warped product manifolds in terms of curvatures of their submanifolds. We also state the geodesics equation.
\end{abstract}

\section{Introduction}
\label{sec0}

Singly warped products were first introduced by R. L. Bishop and B. O'Neil \cite{Ref1} in their attempt to construct a class of Riemannian manifolds with negative curvature. O'Neil also studied Robertson Walker, Schwarschild and Kruskal space-times as warped products in \cite{Ref1a}. Later Beem, Ehrlich and Powell pointed out that many exact solutions to Einstein's field equation can be expressed in terms of Lorentzian warped products in \cite{Ref2}. Since then warped product spaces play a crucial role in a plethora of physical applications such as general relativity, string and supergravity theories \cite{Ref3}.
\par In the present work, for pseudo-Riemannian doubly warped product manifolds, we prove expressions that relate the Riemann, Ricci and scalar curvatures with those of their submanifolds. The derivation is coordinate independent and we also give the equation of geodesics. Finally, in the Appendix, we write the components of the curvatures in a local coordinate system.

\section{Preliminaries}
\label{sec1}

In this section adopting the notations of \cite{Ref1a} we recall briefly basic notions of product Riemannian manifolds and give the definition of the doubly warped product space.  
\par Let $(B,g_B)$ and $(F,g_F)$ be m and n-dimensional pseudo-Riemannian manifolds respectively. Then the $M=B\times F$ is an $(m+n)$-dimensional pseudo-Riemannian manifold with $\pi:\, M\rightarrow B$ and $\sigma: \, M\rightarrow F$ the usual smooth projection maps.
\par We use the natural product coordinate systems on the product manifold $M$. If $(p_0,q_0)\in M$ and $(U_B,x)$, $(U_F,y)$ are coordinate charts on $B$ and $F$ such that $p_0\in B$ and $q_0\in F$, then we can define a coordinate chart $(U_M,z)$ on $M$ such that $U_M$ is an open subset in $M$ contained in $U_B\times U_F$, $(p_0,q_0)\in U_M$ and $\forall (p,q)\in U_M$, 
$z(p,q)=(x(p),y(q))$ where $x(p)=(x^1(p),\cdots,x^m(p))$ and $y(q)=(y^{m+1}(q), \cdots, y^{m+n}(q))$.
\par The set of all smooth and positive valued functions $f: \, B\rightarrow R^{+}$ is denoted by $\mathcal{F}(B)=C^{\infty}(B)$. The lift of $f$ to $M$ is defined by $\tilde{f}=f\circ \pi\in \mathcal{F}(M)$.
\par If $x_p\in T_p(B)$ and $q\in F$ then the lift $\tilde{x}_{(p,q)}$ of $x_p$ to $M$ is the unique tangent vector in $T_{(p,q)}(B\times \{q\})$ such that $d\pi_{(p,q)}(\tilde{x}_{(p,q)})=x_p$ and $d\sigma_{(p,q)}(\tilde{x}_{(p,q)})=0$. The set of all such horizontal tangent vector lifts will be denoted by $L_{(p,q)}(B)$.
\par Moreover we can define lifts of vector fields. Let $X\in \Hr(B)$ where $\Hr(B)$ is the set of smooth vector fields, then the lift $\tilde{X}$ of $X$ to $M$ is the unique element of $\Hr(M)$ whose value at each $(p,q)$ is the lift of $X_p$ to $(p,q)$. The set of such lifts will be denoted by $\mathcal{L}_{(p,q)}(B)$.
\begin{definition}
Let $(B,g_B)$ and $(F,g_F)$ be pseudo-Riemannian manifolds and $f:\, B\rightarrow \mathbb R^+$, $h:\, F\rightarrow \mathbb R^+$ be smooth functions. The doubly warped product space is the product manifold furnished with the metric tensor defined by
\beqr
g_M=(h\circ \sigma)^2 \pi^*(g_B)+(f\circ \pi)^2 \sigma^*(g_F).
\label{sec1 : eq1}
\feqr
Explicitly if $v$ is tangent to $B\times F$ at $(p,q)$ then 
\beqr
<v,v>=h^2(q)<d\pi(v),d\pi(v)>+f^2(p)<d\sigma(v),d\sigma(v)>.
\label{sec1 : eq2}
\feqr
We will denote this structure by $M=B\, _h\!\times_f F$. 
\label{sec1 : def1} 
\end{definition}
The warped metric \rf{sec1 : eq1} is characterized by:
\begin{description}
\item[$(\alpha)$] For each $q\in F$ the map $\pi\lceil(B\times \{q\})$ is a positive homothety onto $F$ with scale factor $1/h(q)$.
\item[$(\beta)$] For each $p\in B$ the map $\sigma\lceil(\{p\}\times F)$ is a positive homothety onto $B$ with scale factor $1/f(p)$.
\item[$(\gamma)$] For each $(p,q)\in M$, the leaf $B\times \{q\}$ and the fiber $\{p\}\times F$ are orthogonal at $(p,q)$.
\end{description}
If $h=1$ and $f\neq 1$ then we obtain a singly warped product manifold. If $h=f=1$ then we have a product manifold.
\section{Covariant derivatives}
\label{sec2}

As in the case of a pseudo-Riemannian product manifold we can define the orthogonal projections:
\beqr
\textrm{tan}: \, T_{(p,q)}(M)\rightarrow T_{(p,q)}(\{p\}\times F), \non \\
\textrm{nor}: \, T_{(p,q)}(M)\rightarrow T_{(p,q)}(B\times\{q\})
\label{sec2 : eq1}
\feqr 
and thus vectors tangent to fibers $\{p\}\times F=\pi^{-1}(p)$ are vertical while vectors tangent to leaves $B\times \{q\}=\sigma^{-1}(q)$ are horizontal.
The lifts of functions and vector fields will be denoted without the tilde for simplicity. Also all geometrical quantities with a superscript $B$ (or $F$) are the pullbacks by $\pi$ (or $\sigma$) of the corresponding ones on $B$ (or $F$).
\begin{proposition}
On $M=B\, _h\!\times_f F$, if $X,Y \in \mathcal{L}(B)$ and $V,W \in \mathcal{L}(F)$ then 
\begin{description}
\item[$(1)$] \begin{displaymath} \textrm{nor} \, {}^M \! D_X Y= {}^B \! D_X Y, \quad \textrm{nor} \,{}^M \! D_X Y\in\mathcal{L}(B) \end{displaymath} where $\textrm{nor} \,{}^M \! D_X Y$ is the lift of the Levi-Civita connection ${}^B \! D_X Y$ to $B$.
\item[$(2)$] \begin{displaymath} \textrm{tan} \, {}^M \! D_X Y= II(X,Y)=-\frac{<X,Y>}{h} \textrm{grad} \, h \end{displaymath} where the first fundamental form is defined by $g_M(X,Y)=<X,Y>$ and $II(X,Y)$ is the shape tensor (or second fundamental form) of $B$  to $M$ \footnote{In this relation $h$ stands for $\tilde{h}=h\circ \sigma$ and $grad\, h$ for $grad\, \tilde{h}=\widetilde{grad\, h}$.}.
\item[$(3)$] \begin{displaymath} {}^M D_X V={}^M D_V X, \, \, nor {}^M D_X V=\frac{V h}{h} X, \, \, tan {}^M D_V X=\frac{X f}{f} V \end{displaymath}
\item[$(4)$] \begin{displaymath} \textrm{nor} \, {}^M \! D_V W= II(V,W)=-\frac{<V,W>}{f} \textrm{grad} f \end{displaymath} where $II(V,W)$ is the shape tensor of $F$.
\item[$(5)$] \begin{displaymath} \textrm{tan} \, {}^M \! D_V W= {}^F \! D_V W, \quad \textrm{tan} \, {}^M \! D_V W\in \mathcal{L}(F). \end{displaymath}
\end{description}
\label{sec2 : prop1}
\end{proposition}
\begin{proof}
\label{sec2 : proof1} 
\end{proof} 
\begin{description}
\item[$(1)$] The vector fields $X, Y$ are tangent to the leaves. On a leaf, $\textrm{nor} \, {}^M \! D_X Y$ is the leaf covariant derivative applied to the restrictions of $X$ and $Y$ to that leaf. Then $\pi$-relatedness follows since homotheties preserve Levi-Civita connections. 
\item[$(2)$] The Koszul formula \footnote{The Koszul formula for a pseudo-Riemannian manifold $M$ reads \begin{displaymath} 2<D_X Y,Z>=X<Y,Z>-Z<X,Y>+Y<Z,X>-<X,[Y,Z]>+<Z,[X,Y]>+<Y,[Z,X]>\end{displaymath} for $X,Y,Z \in \Hr(M)$.} reduces to 
\beqr
2<{}^M D_X Y,V>=-V<X,Y>
\label{sec2 : eq2}
\feqr 
since $<X,V>=<Y,V>=0$, $[X,V]=[Y,V]=0$ and $<V,[X,Y]>=0$. Also
\beqr
V<X,Y>=V[h^2(<X,Y>\circ \pi)]=2\left(\frac{Vh}{h}\right)<X,Y>=2<\frac{<X,Y> grad \, h}{h},V>. 
\label{sec2 : eq3}
\feqr
From \rf{sec2 : eq2} and \rf{sec2 : eq3} we end up with the desired result. 
\item[$(3)$] Since $[X,V]=0$ we have ${}^M D_X V={}^M D_V X$. Using the Koszul formula twice for $<{}^M D_XV,Y>$ and $<{}^M D_X V,U>$ we obtain 
\beqr
2<{}^M D_XV,Y>&=& 2<\textrm{nor}\,{}^MD_X V,Y>=V<X,Y>=2<\frac{V h}{h}X, Y> \non \\
2<{}^M D_X V,W>&=& 2<\textrm{tan}\,{}^MD_X V,W>=X<V,W>=2<\frac{Xf}{f}V,W>.
\label{sec2 : eq4}
\feqr
\item[$(4)$] The proof is identical to (2) for $V,W \in \mathcal{L}(F)$.
\item[$(5)$] This proof is similar to (1) for vector fields tangent to fibers.
\end{description}
\section{Riemann, Ricci and scalar curvatures}
\label{sec3} 

The function $R: \, \Hr^3(M)\rightarrow \Hr(M)$ given by
\beqr
R_{XY}Z=D_{[X,Y]}Z-[D_X,D_Y]Z
\label{sec3 : eq0}
\feqr
is a $(1,3)$ tensor field on $M$ the so called Riemann curvature of $M$. Sometimes we use the symbols $R(X,Y)Z$ or $R(X,Y,Z)$ instead of $R_{XY}Z$.
\begin{proposition}
On $M=B\, _h\!\times_f F$, if $X,Y,Z \in \mathcal{L}(B)$ and $U,V,W \in \mathcal{L}(F)$ then 
\begin{description}
\item[$(1)$] \begin{displaymath} {}^M R_{XY}Z={}^B R_{XY}Z- \frac{\parallel grad\, h\parallel^2}{h^2}\left(<X,Z>Y-<Y,Z>X\right)\end{displaymath} where ${}^M R_{XY}Z$ is the lift of the Riemann curvature tensor ${}^B R_{XY}Z$ to $M$.
\item[$(2)$] \begin{displaymath} {}^M R_{VX}Y=\frac{H^f(X,Y)}{f}V+\frac{<X,Y>}{h}{}^FD_V (grad\, h)={}^M R_{VY}X \end{displaymath} where $H^f$ is the Hessian of $f$ \footnote{The Hessian $H^f$ of $f\in \mathcal{F}(M)$ is the symmetric $(0,2)$ tensor field such that \begin{displaymath} H^f(X,Y)=XYf-(D_XY)f=<D_X(grad \,f),Y>.\end{displaymath}}.
\item[$(3)$] \begin{displaymath} {}^M R_{XY}V=\frac{Vh}{h}\left[\left(\frac{Xf}{f}\right)Y-\left(\frac{Yf}{f}\right)X\right], \end{displaymath}
\begin{displaymath} {}^M R_{VW}X=\frac{Xf}{f} \left[\left(\frac{Vh}{h}\right)W-\left(\frac{Wh}{h}\right)V\right]\end{displaymath}
\item[$(4)$] \begin{displaymath} {}^M R_{XV}W={}^M R_{XW}V=\frac{H^h(V,W)}{h}X+\frac{<V,W>}{f} {}^BD_X (grad\, f) \end{displaymath} where $H^h$ is the Hessian of $h$.
\item[$(5)$] \begin{displaymath} {}^M R_{VW}U={}^F R_{VW} U-\frac{\parallel grad\, f\parallel^2}{f^2}\left( <V,U>W-<W,U>V\right). \end{displaymath}
\end{description}
\label{sec3 : prop1}
\end{proposition}
\begin{proof}
\label{sec3 : proof1} 
\end{proof} 
\begin{description}
\item[$(1)$] For the vector fields $X,Y,Z,A\in \mathcal{L}(B)$ we have
\beqr
\label{sec3 : eq01}
&& <{}^M R_{XY}Z,A>= <{}^B R_{XY}Z,A>-<II(X,Z),II(Y,A)>+<II(X,A),II(Y,Z)> \non \\
&=& <{}^B R_{XY}Z,A>-\frac{\parallel grad \, h\parallel^2>}{h^2}\left(<<X,Z>Y-<Y,Z>X,A>\right)
\feqr 
\item[$(2)$] Since $[X,V]=0$ from definition \rf{sec3 : eq0} we find 
\beqr
{}^M R_{VX}Y=-{}^M D_V({}^M D_X Y)+{}^M D_X({}^M D_VY).
\label{sec3 : eq1}
\feqr
The first term on the right hand side of \rf{sec3 : eq1}, with the help of Proposition \rf{sec2 : prop1}, gives
\beqr
&& {}^MD_V({}^MD_X Y)= {}^MD_V(nor \, {}^MD_XY+tan \, {}^MD_XY)\non \\
&=& \frac{Vh}{h}{}^B D_XY+\frac{({}^BD_XY)f}{f}V-\frac{Vh}{h^2}<X,Y> grad\, h -\frac{<X,Y>}{h}{}^FD_V(grad \, h)
\label{sec3 : eq2}
\feqr
while the second term becomes
\beqr
{}^MD_X({}^MD_VY)=\frac{Vh}{h}\left({}^BD_XY-\frac{<X,Y>}{h}grad\, h\right)+\frac{XYf}{f}V.
\label{sec3 : eq3}
\feqr
Substituting \rf{sec3 : eq2} and \rf{sec3 : eq3} into \rf{sec3 : eq1} we recover the desired relation. 
\item[$(3)$] Assuming that $[X,Y]=0$ (as, for example, for coordinate vector fields) we first prove that
\beqr
{}^MD_X(D_YV)=\frac{XYf}{f}V+\left(\frac{Yf}{f}\right)\left(\frac{Vh}{h}\right)X+\frac{Vh}{h}{}^BD_YX
\label{sec3 : eq4}
\feqr 
and a similar expression for ${}^MD_Y(D_XV)$ with $X$ and $Y$ interchanged. The claim is then straightforward by \rf{sec3 : eq0}. Also note that $<{}^MR_{XY}V,W>=<{}^MR_{VW}X,Y>=0$. 
\item[$(4)$] The proof is established following steps similar to those in (2).
\item[$(5)$] This relation is justified as in (1).
\end{description}
The Ricci tensor relative to a frame field \footnote{A frame field is a set $\{E_i\}, \, i=1,\cdots, dimM$ of orthonormal vector fields.} is defined by
\beqr
Ric(X,Y)=\sum_{i}\epsilon_i <R_{XE_i}Y,E_i>
\label{sec3 : eq5}
\feqr
where $\epsilon_i=<E_i,E_i>$.
\begin{corollary}
On $M=B\, _h\!\times_f F$, if $X,Y \in \mathcal{L}(B)$, $V,W \in \mathcal{L}(F)$, $m=dim B>1$ and $n=dim F>1$ then
\begin{description}
\item[$(1)$] \begin{displaymath} {}^MRic(X,Y)={}^B Ric(X,Y)-<X,Y>\left[\frac{{}^F\Delta h}{h}+(m-1)\frac{\parallel grad\,h\parallel^2}{h^2}\right]-\frac{n}{f}H^f(X,Y)\end{displaymath} 
\item[$(2)$] \begin{displaymath} {}^MRic(X,V)=(m+n-2)\left(\frac{Xf}{f}\right)\left(\frac{Vh}{h}\right) \end{displaymath}
\item[$(3)$] \begin{displaymath} {}^MRic(V,W)={}^F Ric(V,W)-<V,W>\left[\frac{{}^B\Delta f}{f}+(n-1)\frac{\parallel grad\,f\parallel^2}{f^2}\right]-\frac{m}{h}H^h(V,W)\end{displaymath}
\end{description}
\label{sec3 : cor1}
\end{corollary}
\begin{proof}
\label{sec3 : proof2} 
\end{proof} 
Let $\{{}^BE^{i}\}, \, i=1,\cdots,dim B=m$ be a frame field on an open set $U_B\subseteq B$ and $\{{}^FE^{i}\}, \, i=m+1,\cdots,m+dim F=m+n$ be a frame field on an open set $U_F\subseteq F$, then $\{{}^ME^{i}\}, \, i=1,\cdots,dim M=m+n$ be a frame field on the open set $U_M\subseteq U_B\times U_F\subseteq B\times F$.
\begin{description}
\item[$(1)$] Using definition \rf{sec3 : eq5} we have 
\beqr
&& {}^MRic(X,Y)= \sum_{i=1}^{m+n}{}^M\epsilon_{i}<{}^MR_{X{}^ME_i}Y,{}^ME_i> \non \\
&=&\sum_{i=1}^{m}{}^B\epsilon_{i}<{}^MR_{X{}^BE_i}Y,{}^BE_i>+\sum_{i=m+1}^{m+n}{}^F\epsilon_{i}<{}^MR_{X{}^FE_i}Y,{}^FE_i> \non \\
&=& \sum_{i=1}^{m}{}^B\epsilon_{i}\left[<{}^B R_{X{}^BE_i}Y,{}^BE_i>-\frac{\parallel grad \, h\parallel^2}{h^2}\left(<<X,Y>{}^BE_i-<{}^BE_i,Y>X,{}^BE_i>\right)\right]\non \\
&-& \sum_{i=m+1}^{m+n}{}^F\epsilon_{i}\left[\frac{H^f(X,Y)}{f}<V,{}^FE_i>+\frac{<X,Y>}{h}<{}^FD_{{}^FE_i} (grad\, h),{}^FE_i>\right]\non \\
&=& {}^B Ric(X,Y)-\frac{\parallel grad \, h\parallel^2}{h^2}<X,Y>\left(\sum_{i=1}^{m}{}^B\epsilon_i^2-1\right) \non \\
&-& \frac{H^f(X,Y)}{f}\sum_{i=m+1}^{m+n}{}^F\epsilon_i<{}^FE_i,{}^FE_i>-\frac{<X,Y>}{h}\sum_{i=m+1}^{m+n}{}^F\epsilon_i <{}^FD_{{}^FE_i}(grad\,h),{}^FE_i> 
\label{sec3 : eq6}
\feqr
from which taking into account that $\sum_{i=m+1}^{m+n}{}^F\epsilon_i <{}^FD_{E_i}(grad\,h),{}^FE_i>={}^F\Delta h$ we recover the known result.   
\item[$(2)$] Using the definition of Ricci tensor we have
\beqr
&& {}^MRic(X,V)=\sum_{i=1}^{m}{}^B\epsilon_i<{}^BR_{X{}^BE_i}V,{}^BE_i>+\sum_{i=m+1}^{m+n}{}^F\epsilon_i<{}^FR_{V{}^FE_i}X,{}^FE_i> \non \\
&=& \left(\frac{Vh}{h}\right)\left[\left(\frac{Xf}{f}\right)\sum_{i=1}^{m}{}^B\epsilon_i^2-\frac{1}{f}\left(\sum_{i=1}^{m}{}^B\epsilon<X,{}^BE_i>{}^BE_i\right)f\right]\non \\
&+&\left(\frac{Xf}{f}\right)\left[\left(\frac{Vh}{h}\right)\sum_{i=m+1}^{m+n}{}^F\epsilon_i^2-\frac{1}{h}\left(\sum_{i=m+1}^{m+n}{}^F\epsilon<V,{}^FE_i>{}^FE_i\right)h\right]\non \\
&=& (m+n-2)\left(\frac{Xf}{f}\right)\left(\frac{Vh}{h}\right).
\label{sec3 : eq7}
\feqr
\item[$(3)$] The proof is similar to (1).
\end{description}
The scalar curvature $R$ of $M$ is the contraction of its Ricci tensor which relative to a frame field yields
\beqr
R=\sum_{i\neq j} K(E_i,E_j)=2\sum_{i<j}K(E_i,E_j), \quad \textrm{where} \quad K(X,Y)=\frac{<R_{XY}U,V>}{<X,X><Y,Y>-<X,Y>^2}
\label{sec3 : eq8}
\feqr 
is the sectional curvature.
\begin{corollary}
On $M=B\, _h\!\times_f F$, if $X,Y \in \mathcal{L}(B)$, $V,W \in \mathcal{L}(F)$, $m=dim B>1$ and $n=dim F>1$ then
\beqr
{}^M R &=& \frac{{}^B R}{h^2}-2m\frac{{}^F\Delta h}{h}-m(m-1)\frac{\parallel grad\,h\parallel^2}{h^2}\non \\
&+&\frac{{}^F R}{f^2} - 2n\frac{{}^B\Delta f}{f}-n(n-1)\frac{\parallel grad\,f\parallel^2}{f^2}
\label{sec3 : eq9}
\feqr
\label{sec3 : cor2}
\end{corollary}
If we change the sign in the definition of Riemann curvature then its components change sign but the Ricci tensor and scalar remain intact. 
\section{Geodesics}
\label{sec4}

A curve $\gamma: \, I \rightarrow M$ can be written as $\gamma(s)=(\alpha(s),\beta(s))$ with $\alpha$, $\beta$ the projections of $\gamma$ into $B$ and $F$.
\begin{proposition}
A curve $\gamma=(\alpha,\beta)$ in $M$ is geodesic iff
\begin{description}
\item[$(1)$] \begin{displaymath} \beta^{''}=<\alpha^{'},\alpha^{'}>h\circ \beta grad\, h-\frac{2}{(f\circ \alpha)}\frac{d}{ds}(f\circ \alpha) \, \beta^{'}. \end{displaymath} 
\item[$(2)$] \begin{displaymath} \alpha^{''}=<\beta^{'},\beta^{'}>f\circ \alpha grad\, f-\frac{2}{h\circ \beta}\frac{d}{ds}((h\circ \beta)) \, \alpha^{'}. \end{displaymath}  
\end{description}
\label{sec4 : prop1} 
\end{proposition}
\begin{proof}
\label{sec4 : proof1} 
\end{proof} 
\begin{description}
\item[$(1)$] Consider the case when $\gamma^{'}(0)$ is neither horizontal nor vertical. Then $\alpha$, $\beta$ are regular, namely $\forall s\in I, \, \alpha^{'}(s), \beta^{'}(s)\neq 0$, and thus locally integral curves. This means that $\alpha^{'}(s)=X_{\alpha(s)}$ and $\beta^{'}(s)=V_{\beta(s)}$ for $X\in \mathcal{L}(B)$ and $V\in\mathcal{L}(F)$. Moreover $\gamma$ is an integral curve and
\beqr
\gamma^{''}=D_{X+V}(X+V)=D_XX+D_XV+D_VX+D_VV.
\label{sec4 : eq1}
\feqr
The curve $\gamma$ is geodesic iff $\gamma^{''}=0$ or equivalently iff $\tan \, \gamma^{''}=\rm{nor} \, \gamma^{''}=0$. Using Proposition (3.1) in \rf{sec4 : eq1} we obtain
\beqr
\textrm{tan}\left(D_XX+D_XV+D_VX+D_VV \right)&=& 0 \quad \Rightarrow \non \\
-\frac{<X,X>}{h}\textrm{grad}\, h+2\frac{Xf}{f}V+{}^FD_VV&=& 0 
\label{sec4 : eq2}
\feqr
which proves (1). The second identity is reproduced by taking the $\textrm{nor}$ component of \rf{sec4 : eq1}.
\item[$(2)$] If $\gamma^{'}(0)$ is horizontal or vertical and nonzero then the geodesic $\gamma$ does not remain in the leaves $\sigma^{-1}(q)$ or the fibers $\pi^{-1}(p)$. Hence there is a sequence $\{s_i\}\rightarrow 0$ such that $\forall i$, $\gamma^{'}(s_i)$ is neither horizontal nor vertical. Then $(1)$ and $(2)$ follow by continuity of Case 1. 
\end{description}


\addcontentsline{toc}{subsection}{Appendix}
\section*{Appendix}
\label{apA}
\renewcommand{\theequation}{A.\arabic{equation}}
\setcounter{equation}{0}
Let $\{x^{\mu}\}, \, \mu=1,\cdots,dim B=m$ be a local coordinate system on an open set $U_B\subseteq B$ and $\{y^{\alpha}\}, \, \alpha=m+1,\cdots,m+dim F=m+n$ be a local coordinate system on an open set $U_F\subseteq F$, then $\{x^{\hat{\mu}}\}, \, \hat{\mu}=1,\cdots,dim M=m+n$ is a local coordinate system on the open set $U_M\subseteq U_B\times U_F\subseteq B\times F$. The components of the Levi-Civita connections are
\beqr
\textrm{nor} \, {}^M \! D_{\partial_{\mu}} (\partial_{\nu})&=& {}^B \! D_{\partial_{\mu}} (\partial_{\nu})={}^B\Gamma^{\rho}_{\mu \nu} \partial_{\rho},
\label{A :eqa1} \\
\textrm{tan} \, {}^M \! D_{\partial_{\mu}} (\partial_{\nu})&=&-\frac{<\partial_{\mu},\partial_{\nu}>}{h} grad \, h=-\frac{h}{f^2}g_{\mu \nu}g^{\alpha \beta} \partial_{\beta}h \, \partial_{\alpha}=\Gamma^{\alpha}_{\mu \nu}\partial_{\alpha},
\label{A :eqa2} \\
\textrm{nor} {}^MD_{\partial_{\mu}}( \partial_{\alpha})&=&\frac{\partial_{\alpha}h}{h}\delta^{\nu}_{\mu}\partial_{\nu}=\Gamma^{\nu}_{\mu \alpha} \partial_{\nu}, \,\, \textrm{tan} {}^M D_{\partial_{\alpha}} (\partial_{\mu})=\frac{\partial_{\mu}f}{f}\delta^{\beta}_{\alpha} \partial_{\beta}=\Gamma^{\beta}_{\alpha \mu} \partial_{\beta},
\label{A :eqa3} \\
\textrm{nor} \, {}^M \! D_{\partial_{\alpha}} (\partial_{\beta})&=&-\frac{f}{h^2}g_{\alpha \beta}g^{\mu \nu} \partial_{\nu}f \, \partial_{\mu}=\Gamma^{\mu}_{\alpha \beta} \partial_{\mu},
\label{A :eqa4}\\
\textrm{tan} \, {}^M \! D_{\partial_{\alpha}} (\partial_{\beta})&=& {}^F\Gamma^{\gamma}_{\alpha \beta} \partial_{\gamma}.
\label{A :eqa5}
\feqr
For coordinate vector fields, $R_{\partial_{\hat{\lambda}}\partial_{\hat{\rho}}}(\partial_{\hat{\nu}})=R^{\hat{\mu}}_{\,\,\,\hat{\nu} \hat{\lambda} \hat{\rho}}\partial_{\hat{\mu}}$ and the components of the Riemann curvature are given by
\beqr
{}^MR^{\mu}_{\,\,\,\nu \lambda \rho}&=&{}^BR^{\mu}_{\,\,\,\nu \lambda \rho}-\frac{\parallel\partial_{\alpha}h\parallel^2}{f^2}\left(\delta^{\mu}_{\lambda}g_{\nu \rho}-\delta^{\mu}_{\rho}g_{\nu \lambda}\right),
\label{A : eq6} \\
{}^MR^{\alpha}_{\,\,\, \mu \beta \nu}&=& -\frac{\delta^{\alpha}_{\beta}}{f}D_{\nu}(\partial_{\mu}f) -\frac{h}{f^2}g_{\mu \nu}D_{\beta}(\partial^{\alpha}h),
\label{A : eq7} \\
{}^MR^{\alpha}_{\,\,\,\mu \beta \gamma}&=& \partial_{\mu}(\ln f)\left[\delta^{\alpha}_{\beta}\partial_{\gamma}(\ln h)-\delta^{\alpha}_{\gamma}\partial_{\beta}(\ln h)\right],
\label{A : eq71}\\
{}^MR^{\mu}_{\,\,\,\alpha \nu \lambda}&=& \partial_{\alpha}(\ln h)\left[\delta^{\mu}_{\nu}\partial_{\lambda}(\ln f)-\delta^{\mu}_{\lambda}\partial_{\nu}(\ln f)\right],
\label{A : eq72}\\
{}^MR^{\alpha}_{\,\,\,\mu \nu \lambda}&=&\frac{1}{2f^2}\partial^{\alpha}(h^2)\left[g_{\mu \lambda}\partial_{\nu}(\ln f)-g_{\mu \nu}\partial_{\lambda}(\ln f)\right],
\label{A : eq8} \\
{}^MR^{\mu}_{\,\,\,\alpha \beta \gamma}&=&\frac{1}{2h^2}\partial^{\mu}(f^2)\left[g_{\alpha \gamma}\partial_{\beta}(\ln h)-g_{\alpha \beta}\partial_{\gamma}(\ln h)\right],
\label{A :eq81}\\
{}^MR^{\mu}_{\,\,\, \alpha \nu \beta}&=& -\frac{\delta^{\mu}_{\nu}}{h}D_{\beta}(\partial^{\alpha}h) -\frac{f}{h^2}g_{\alpha \beta}D_{\nu}(\partial^{\mu}f),
\label{A : eq9} \\
{}^MR^{\alpha}_{\,\,\,\beta \gamma \epsilon}&=&{}^FR^{\alpha}_{\,\,\,\beta \gamma \epsilon}-\frac{\parallel\partial_{\mu}f\parallel^2}{h^2}\left(\delta^{\alpha}_{\gamma}g_{\beta \epsilon}-\delta^{\alpha}_{\epsilon}g_{\beta \gamma}\right).
\label{A : eq10} 
\feqr
The components of the Ricci tensor are
\beqr
{}^MR_{\mu \nu}&=& {}^BR_{\mu \nu}-\frac{g_{\mu \nu}}{f^2}\left[h {}^F\Delta h+(m-1)\parallel\partial_{\alpha} h\parallel^2\right]-\frac{n}{f}D_{\mu}(\partial_{\nu}f),
\label{A : eq11}\\
{}^MR_{\mu \alpha}&=&(m+n-2)\partial_{\mu}(\ln f)\partial_{\alpha}(\ln h),
\label{A : eq12}\\
{}^MR_{\alpha \beta}&=& {}^BR_{\alpha \beta}-\frac{g_{\alpha \beta}}{h^2}\left[f{}^B\Delta f+(n-1)\parallel\partial_{\mu} f\parallel^2\right]-\frac{m}{h}D_{\alpha}(\partial_{\beta}h)
\label{A : eq13}
\feqr
and the Ricci scalar is given by
\beqr
{}^M R &=& \frac{{}^B R}{h^2}-2m\frac{{}^F\Delta h}{hf^2}-m(m-1)\frac{\parallel\partial_{\alpha}(\ln h)\parallel^2}{f^2}\non \\
&+&\frac{{}^F R}{f^2} - 2n\frac{{}^B\Delta f}{fh^2}-n(n-1)\frac{\parallel\partial_{\alpha}(\ln f)\parallel^2}{h^2}.
\label{A : eq14}
\feqr
All the expressions of the Appendix have also been recorded in \cite{Ref4} and \cite{Ref5}. 
\bibliographystyle{plain}

\end{document}